\documentclass[sigconf]{acmart}

\usepackage{balance}
\usepackage{multirow}

\newcommand{\E}{\mathop{\textrm{E}}}
\newcommand{\T}{\textsf{T}}
\newcommand{\q}{\textbf{s}}

\AtBeginDocument{%
  \providecommand\BibTeX{{%
    \normalfont B\kern-0.5em{\scshape i\kern-0.25em b}\kern-0.8em\TeX}}}
    
\acmYear{2022}


\copyrightyear{2022}
\acmYear{2022}
\setcopyright{acmlicensed}\acmConference[ICTIR '22]{Proceedings of the 2022
ACM SIGIR International Conference on the Theory of Information Retrieval}{July
11--12, 2022}{Madrid, Spain}
\acmBooktitle{Proceedings of the 2022 ACM SIGIR International Conference on
the Theory of Information Retrieval (ICTIR '22), July 11--12, 2022, Madrid,
Spain}
\acmPrice{15.00}
\acmDOI{10.1145/3539813.3545119}
\acmISBN{978-1-4503-9412-3/22/07}

\begin{document}

\title{A General Framework for Pairwise Unbiased Learning to Rank}

\author{Alexey Kurennoy}
\email{alexey.kurennoy@zalando.ie}
\affiliation{
	\institution{Zalando}
	\city{Dublin}
	\country{Ireland}}
	
\author{John Coleman}
\email{john.coleman@zalando.ie}
\affiliation{
	\institution{Zalando}
	\city{Dublin}
	\country{Ireland}}
	
\author{Ian Harris}
\email{ian.harris@zalando.ie}
\affiliation{
	\institution{Zalando}
	\city{Dublin}
	\country{Ireland}}
	
\author{Alice Lynch}
\email{alice.lynch@zalando.ie}
\affiliation{
	\institution{Zalando}
	\city{Dublin}
	\country{Ireland}}
	
\author{Oisin Mac Fhearai}
\email{oisin.mac.fhearai@zalando.ie}
\affiliation{
	\institution{Zalando}
	\city{Dublin}
	\country{Ireland}}
	
\author{Daphne Tsatsoulis}
\email{daphne.tsatsoulis@zalando.ie}
\affiliation{
	\institution{Zalando}
	\city{Dublin}
	\country{Ireland}}

\renewcommand{\shortauthors}{Alexey Kurennoy et al.}

\begin{abstract}
	Pairwise debiasing is one of the most effective strategies in reducing position bias in learning-to-rank (LTR) models. However, limiting the scope of this strategy, are the underlying assumptions required by many pairwise debiasing approaches. In this paper, we develop an approach based on a minimalistic set of assumptions that can be applied to a much broader range of user browsing patterns and arbitrary presentation layouts. We implement the approach as a simplified version of the Unbiased LambdaMART and demonstrate that it retains the underlying unbiasedness property in a wider variety of settings than the original algorithm. Finally, using simulations with "golden" relevance labels, we will show that the simplified version compares favourably with the original Unbiased LambdaMART when the examination of different positions in a ranked list is not assumed to be independent.
\end{abstract}


\maketitle

\section{Introduction}

Machine learning models for information retrieval and recommendation are typically trained on implicit user feedback (such as clicks, for instance). Implicit feedback has several attractive properties. For example, it is abundant and relatively cheap to obtain. However, it is prone to presentation biases. This means that the implicit feedback corresponding to a certain item depends on the way the item was presented to users during the feedback collection. One important type of presentation bias that is especially pronounced in ranking applications is position bias. This bias arises because user attention is not spread equally between different positions in a ranked list and some of the positions are seen or attract attention more frequently than others.

The position bias renders items that are ranked low by the existing production system as less relevant to users than they really are. Consequently, a machine learning method applied to the collected data tends to mimic the current system. If the current production model is not optimal, its sub-optimality is (at least, partially) passed onto the new model. As a result, making improvements to the existing ranking system becomes more difficult. This has several undesirable implications such as worse user experience, lower revenues, or fairness problems.

There has been a considerable amount of research on ways to eliminate or reduce the position bias in learning-to-rank and recommendation models. Both the offline \cite{Jo:17,Ag:19} and the online \cite{OMa:18} environments have been considered and there exists work aiming to unify the two \cite{OMa:21,Ai:21}. We also refer the reader to a recent survey \cite{Ch:20}. Note that the theme of position bias reduction is different from offline policy evaluation (see \cite{SJo:15**, Li:18, KSa:22} and other references in \cite{KSa}) even though the two domains have some similarities. While offline policy evaluation focuses on assessing the loss that would have been observed under a different policy the goal of debiasing is to estimate the value of the loss that we would have observed if the target signal based on the implicit feedback was not distorted due to the position bias. From this perspective, debiasing is relevant even when the only aim is to evaluate the existing ranking or recommender system (i.\,e. the logging policy).

A lot of papers on the topic of position bias removal, including the seminal work \cite{Jo:17} and its subsequent generalisation \cite{Ag:19}, focused on modifying objective functions that involve summations over individual items in the training data.

Hu et. al. \cite{Hu:19} proposed an alternative approach starting off of a pairwise loss function which is a sum of terms that depend on pairs of items rather than individual items. This approach, called \emph{pairwise debiasing}, gave rise to Unbiased LambdaMART - a state-of-the-art method for unbiased learning-to-rank.

In this paper, we describe a general framework for pairwise unbiased learning-to-rank. In contrast to existing theories, it relies on a smaller and more realistic set of assumptions. Importantly, we do not require that the examination of different positions happen independently. For example, if the user is presented with a number of choices, the independent examination assumption would mean that examining the bottom-most option does not increase in any way the chances that earlier positions have been observed too. However, when users tend to examine the choices in a top-down fashion this property is unlikely to hold.

Furthermore,
we do not assume that the probability of irrelevance and click absence (conditional on item and context features) are proportional at each position. See \cite[Section~4.1.4]{Ai:21} for a discussion of why this assumption is undesirable.

In addition to the above, our framework allows for arbitrary presentation layouts and thus covers both web search where results are typically displayed in a list and e-commerce where users are usually presented with a grid of products.

We demonstrate how the framework can be used in several important contexts to produce unbiased learning-to-rank algorithms.
We also utilise the framework to show that a simplified version of Unbiased LambdaMART maintains the underlying unbiasedness property in a wider range of settings than the original algorithm. We compare the simplified and the original Unbiased LambdaMART in a semi-synthetic experiment and find that the simplified version compares favourably to the original Unbiased LambdaMART when the examination of different positions is not independent.

The contributions of the paper can be summarised as follows.
\begin{itemize}
	\item {\bf Theory} \\ We propose a general pairwise debiasing framework allowing for arbitrary presentation layouts and a broad range of user browsing patterns (including those in which there is dependence in the examination of different positions). To the best of our knowledge, this framework has the weakest set of assumptions to date.

	\item {\bf Methods}
		\begin{itemize} 
			\item We show how the framework can be used to produce unbiased learning-to-rank methods for important types of user behaviour found in e-commerce and web search.
			\item We demonstrate that a simplified version of Unbiased LambdaMART is robust in the sense that the underlying unbiasedness property holds in a broad range of settings.
		\end{itemize}

	\item {\bf Offline Experiments} \\ We conduct an offline semi-synthetic experiment (based on public data with "golden" relevance labels) and find that the simplified version of Unbiased LambdaMART compares favourably with the original algorithm when the examination of different positions does not happen independently.

\end{itemize}

The paper is structured as follows. We discuss related work in Section~\ref{sec:related}. Section~\ref{sec:framework} presents the proposed pairwise debiasing framework. This includes stating its assumptions, formulating a novel unbiased version of the pairwise loss function, and proving its unbiasedness. In Section~\ref{sec:examples}, we demonstrate how our framework can be used in several important settings to produce unbiased learning-to-rank algorithms. Section~\ref{sec:lambda} demonstrates that Unbiased LambdaMART can be modified so that the underlying unbiasedness property holds in a wider range of situations. Finally, Section~\ref{sec:experiments} contains the results of a semi-synthetic experiment in which we compare the performance of the original and the simplified Unbiased LambdaMART.

\section{Related Work}\label{sec:related}

The idea of pairwise debiasing along with the Unbiased LambdaMART method were introduced in \cite{Hu:19}. In addition to the standard examination hypothesis \cite[Section~3.3]{Chu:15} and the positivity of observation propensities, \cite{Hu:19} assumes that clicks on different items happen independently of each other and that the probabilities of irrelevance and click absence are proportional for each position.

A recent work \cite{Sa:20} proposes an unbiased pairwise loss function in the context of collaborative filtering with implicit feedback. It
avoids the assumption about the proportionality between the irrelevance and click absence probabilities but still  
assumes independence between the examination of different positions.
The eye-tracking experiments in \cite{Jo:17*} suggest that people generally view web search results from top to bottom which makes the reliance on the assumption about the examination independence undesirable. When the results are observed in a top-to-bottom fashion, the fact that an item down the list has been observed increases the chances that earlier items have been observed too and hence, the independence of examination indicators cannot hold. As will be seen in Section~\ref{sec:examples}, our framework encompasses \cite{Sa:20} as a special case.

Guo et. al. \cite{Gu:20} build an unbiased learning-to-rank method focusing on the context of e-commerce. In this paper, we develop a unified approach that can tackle both the grid-based e-commerce domain and list-based web search scenarios (see examples in Section~\ref{sec:examples}).

\section{Proposed Framework}\label{sec:framework}

In this section, we formulate a pairwise loss which is unbiased under only two assumptions: the examination hypothesis \cite[Section~3.3]{Chu:15} and the positivity of observation propensities. By the latter, we mean that each of the positions in the layout is observable by users (i.e. there are no positions that can never be examined) and that there are no pairs of positions that can never be examined jointly.

In learning-to-rank, the data is comprised of collections of items. Such collections can be, for example, lists of links returned by a web-search engine, grids of products in an online shop, or the set of elements of a recommendation carousel. We will adopt a common notation and use the letter $q$ to denote a single collection of items from the data. This notation is likely to stem from the fact that each of the collections often has an associated query-string but this is not always the case and plays no role in the context of this research. The set of all item collections in the data will be denoted by $Q$. Furthermore, let $n_q$ stand for the number of items in collection $q\in Q$ and let $x_{q,\,i}$ denote the feature vector associated with the $i$-th item of collection $q\in Q$. The feature vector can include both the item attributes and the properties of the context (such as user and query-string features). Finally, let $r_{q,\,i}$ be the relevance of the item. The item relevance is generally unobserved. Instead, the data contains a target label $c_{q,\,i}$ which is based on the user actions and can be, for example, a click indicator. For the ease of exposition, we will assume that both the relevance and the target label are binary but the results can be easily generalised to the case where the relevance and/or the label take more than two values.
	
\subsection{Assumptions}
	We make the following common assumption about the relationship between the relevance and the observed target label (see, \cite{Chu:15}).

\textbf{Assumption 1 {(the examination hypothesis).}}\label{a:examination}
		\textit{
		The target label equals the true relevance if the user has examined the item and is zero otherwise, i.e.
		\begin{equation*}
			c_{q,\,i} = e_{q,\,i} \cdot r_{q,\,i},\quad \forall\,q\in Q \quad \forall\,i=1,\,\ldots,\,n_q,
		\end{equation*}
		where $e_{q,\,i}$ is the examination indicator.
		}

	Note that we do not assume that the examination indicator and the relevance indicator are independent of each other (i.e. we do not require that the probability of a click be the product of the examination and the relevance probabilities). In personal search, for example, the production system would adapt the ranking with each request and place items that are relevant to the user who submitted the request to more visible positions. In this scenario, being relevant increases the chances of being seen and the independence between relevance and examination does not hold.
	
Next, define the (conditional) examination probabilities
\begin{equation}\label{eq:pk=}
	p_{q,\,i} = P\{e_{q,\,i}=1\mid \mathcal{I}_q\},\quad q\in Q,\quad i=1,\,\ldots,\,n_q,
\end{equation}
and the (conditional) joint examination probabilities
\begin{equation}\label{eq:pij=}
	p_{q,\,i,\,j} = P\{e_{q,\,i}e_{q,\,j}=1\mid \mathcal{I}_q\},\quad q\in Q,\quad i,\,j=1,\,\ldots,\,n_q,
\end{equation}
where $\mathcal{I}_q = \{r_{q,\,1},\,\ldots,\,r_{q,\,n_q},\,x_{q,\,1},\,\ldots,\,x_{q,\,n_q}\}$ is a set of variables we will condition upon in our analysis. See Section~\ref{subsec:estimation} below for an example of how the examination probabilities \eqref{eq:pk=} and \eqref{eq:pij=} can be estimated in practice.

\textbf{Assumption 2.}\label{a:positivity}
 	\textit{The examination probabilities defined in \eqref{eq:pk=} and \eqref{eq:pij=} are non-zero.}

\subsection{Unbiased Pairwise Loss}
Let $f$ be a (ranking) model. A pairwise loss function takes the following form
	\begin{eqnarray}\label{eq:loss}
		L_c &=& \sum_{q}\sum_{i,\,j = 1}^{n_q} \ell(f(x_{q,\,i}),\,c_{q,\,i},\,f(x_{q,\,j}),\,c_{q,\,j}).
	\end{eqnarray}
	In other words, it is the sum of terms that correspond to item pairs in the data. Those terms can be typically decomposed as follows
	\begin{multline}\label{eq:ell}
		\ell(f(x_{q,\,i}),\,c_{q,\,i},\,f(x_{q,\,j}),\,c_{q,\,j}) =
		\\
	 =\ell_{1,\,1}(f(x_{q,\,i}),\,f(x_{q,\,j}))c_{q,\,i}c_{q,\,j} \\
		 + \ell_{1,\,0}(f(x_{q,\,i}),\,f(x_{q,\,j}))c_{q,\,i}(1-c_{q,\,j}) \\
		 + \ell_{0,\,1}(f(x_{q,\,i}),\,f(x_{q,\,j}))(1-c_{q,\,i})c_{q,\,j}\\
		 + \ell_{0,\,0}(f(x_{q,\,i}),\,f(x_{q,\,j}))(1-c_{q,\,i})(1-c_{q,\,j}).
	\end{multline}
	Note that the functions $\ell_{1,\,1}$, $\ell_{1,\,0}$, $\ell_{0,\,1}$, and $\ell_{0,\,0}$ depend only on the scores that the model assigns to the two items in the pair. In the case of RankNet \cite{Bu:10}, for example, the four functions are\footnote{In the derivation of the RankNet algorithm, the functions $\ell_{1,\,1}$ and $\ell_{0,\,0}$ are non-zero but the final algorithm ignores pairs with equal target labels which is equivalent to setting the two functions to zero.}
	\begin{equation}\label{eq:ranknet}
		\begin{array}{l}
			\ell_{1,\,1}(f(x_{q,\,i}),\,f(x_{q,\,j})) = 0, \\
			\ell_{1,\,0}(f(x_{q,\,i}),\,f(x_{q,\,j})) = \log\left(1 + e^{-\sigma\cdot(f(x_{q,\,i}) - f(x_{q,\,j}))}\right), \\
			\ell_{0,\,1}(f(x_{q,\,i}),\,f(x_{q,\,j})) = \log\left(1 + e^{-\sigma\cdot(f(x_{q,\,j}) - f(x_{q,\,i}))}\right), \\
			\ell_{0,\,0}(f(x_{q,\,i}),\,f(x_{q,\,j})) = 0.
		\end{array} 
	\end{equation}
	
	To make the notation more concise, we introduce the following two vector functions:
	\begin{equation} \label{eq:zs}
		\begin{array}{l}
			\mathbf{z}(f(x_{q,\,i}),\,f(x_{q,\,j})) = \left(
				\begin{array}{c}
					\ell_{1,\,1}(f(x_{q,\,i}),\,f(x_{q,\,j})) \\
					\ell_{1,\,0}(f(x_{q,\,i}),\,f(x_{q,\,j})) \\
					\ell_{0,\,1}(f(x_{q,\,i}),\,f(x_{q,\,j})) \\
					\ell_{0,\,0}(f(x_{q,\,i}),\,f(x_{q,\,j}))
				\end{array}
			\right), \\
			\q(b_1,\,b_2) = \left(
				\begin{array}{c}
					b_1b_2 \\ 
					b_1(1-b_2) \\ 
					(1-b_1)b_2\\ 
					(1-b_1)(1-b_2)
				\end{array}
			\right), \quad b_1,\,b_2 \in \{0,\,1\}.
		\end{array}
	\end{equation}
The symbols $b_1$ and $b_2$ in the definition of $\textbf{s}$ stand for two binary indicators. Below we will use either click or relevance indicators in their place. Note that the function $\textbf{s}$ one-hot encodes the type of the item pair. To be more specific, let us consider a given item pair $(i,\,j)$ from collection $q\in Q$. If $\textbf{s}$ is computed from click indicators $c_{q,\,i}$ and $c_{q,\,j}$, it equals $(1,\,0,\,0,\,0)^\T$ when both of the items were clicked, $(0,\,1,\,0,\,0)^\T$ when only the first item from the pair was clicked, $(0,\,0,\,1,\,0)^\T$ when only the second item was clicked, and $(0,\,0,\,0,\,1)^\T$ when none of the two items was clicked. Similarly, if $\textbf{s}$ is computed from relevance indicators it gives analogous one-hot encodings but with respect to the relevance of items $i$ and $j$.

	Using the notation introduced in \eqref{eq:zs}, we can rewrite \eqref{eq:loss} and \eqref{eq:ell} as follows:
	\begin{equation}\label{eq:Lc}
		L = \sum_q \sum_{i,\,j=1}^{n_q} \left(\textbf{z}(f(x_{q,\,i}),\,f(x_{q,\,j}))\right)^\T\q(c_{q,\,i},c_{q,\,j}).
	\end{equation}

The loss in \eqref{eq:Lc} (or, equivalently, in \eqref{eq:loss}) is computed from the implicit feedback (click indicators) as opposed to the true relevance indicators since the latter are unavailable (unobserved). Because of the position bias, the optimisation of \eqref{eq:Lc} generally leads to suboptimal algorithms and may reinforce undesirable properties of existing production systems (such as unfairness). Our goal now will be to construct a matrix (denoted by $\textbf{A}_{q,\,i,\,j}$ below) that can convert the term computed from the click indicators into the same term but computed from the relevance indicators under the expectation sign. In other words, this matrix will enable us to utilise the available click data and ``infer'' how items $i$ and $j$ compare in terms of their  true relevance in expectation. Then by injecting this matrix into the loss \eqref{eq:Lc}, we will obtain a new and unbiased loss function. This new loss will be still computed from the clicks but its expected value will equal that of the loss computed from the true (unobserved) relevance values.

We begin by defining a matrix (function) $\textbf{B}$ depending on two examination indicators,
	\begin{equation} \label{eq:B=}
		\textbf{B}(e_1,\,e_2) = \left(\begin{array}{cccc}
			e_1e_2 & 0 & 0 & 0 \\
			e_1(1-e_2) & e_1 & 0 & 0 \\
			(1-e_1)e_2 & 0, & e_2 & 0 \\
			(1-e_1)(1-e_2) & (1-e_1) & (1-e_2) & 1
		\end{array}\right).
	\end{equation}
For any $q\in Q$ and $i,\,j=1,\,\ldots,\,n_q$, the matrix $\textbf{B}(e_{q,\,i},\,e_{q,\,j})$ relates the terms $\textbf{s}(c_{q,\,i},\,c_{q,\,j})$ and $\textbf{s}(r_{q,\,i},\,r_{q,\,j})$. Specifically, (under Assumption~1) it holds that
\begin{equation*}
	\textbf{s}(c_{q,\,i},\,c_{q,\,j}) = \textbf{B}(e_{q,\,i},\,e_{q,\,j})\cdot\q(r_{q,\,i},\,r_{q,\,j}).
\end{equation*}
For example, if both of the two items are relevant ($r_{q,\,i} = 1$, $r_{q,\,j} = 1$) we have that $\textbf{s}(r_{q,\,i},\,r_{q,\,j}) = \textbf{s}(1,\,1) = (1,\,0,\,0,\,0)^\T$. At the same time, if only the first of the two items gets examined ($e_{q,\,i} = 1$, $e_{q,\,j} = 0$), according to the examination hypothesis, we will only observe a click on the first item and $\textbf{s}(c_{q,\,i},\,c_{q,\,j}) = \textbf{s}(1,\,0) = (0,\,1,\,0,\,0)^\T$. It holds that $\textbf{B}(1,\,0)\cdot \textbf{s}(1,\,1) = \textbf{s}(1,\,0)$, i.\,e. the matrix $\textbf{B}(1,\,0)$ produces what is observed ($\textbf{s}(1,\,0)$) from the underlying relevance-based term $\textbf{s}(1,\,1)$.

In practice, of course, we want to achieve the opposite: reconstruct the unobserved true preference $\textbf{s}(r_{q,\,i},\,r_{q,\,j})$ from the click feedback $\textbf{s}(c_{q,\,i},\,c_{q,\,j})$, at least, in expectation. To that end, we take the inverse of the expectation of $\textbf{B}$ and set
\begin{equation*}
	\textbf{A}_{q,\,i,\,j} = \left(\E[\textbf{B}(e_{q,\,i},\,e_{q,\,j})\mid\mathcal{I}_q]\right)^{-1},\quad q\in Q,\quad i,\,j=1,\,\ldots,\,n_q.
\end{equation*}
As will be seen in the proof of Theorem~\ref{th:1} below, this matrix has the desired property, that is it turns $\textbf{s}(c_{q,\,i},\,c_{q,\,j})$ into $\textbf{s}(r_{q,\,i},\,r_{q,\,j})$ under the expectation sign.

Before formulating our unbiased loss and proving its unbiasedness, we note that although the definition of $\textbf{A}_{q,\,i,\,j}$ above involves matrix inversion, this matrix can be computed directly and efficiently from the examination probabilities \eqref{eq:pk=} and \eqref{eq:pij=}. Specifically,

\begin{equation}\label{eq:A=}
	\textbf{A}_{q,\,i,\,j} = \left(
				\begin{array}{cccc}
					a_{q,\,i,\,j} & 0 & 0 & 0 \\
			a_{q,\,i} - a_{q,\,i,\,j} & a_{q,\,i} & 0 & 0 \\
			a_{q,\,j} - a_{q,\,i,\,j} & 0, & a_{q,\,j} & 0 \\
			1 - a_{q,\,i} - a_{q,\,j} + a_{q,\,i,\,j} & 1-a_{q,\,i} & 1-a_{q,\,j} & 1
				\end{array}
			\right),
\end{equation}
where $a_{q,\,i}$ and $a_{q,\,j}$ are the inverses of the individual examination probabilities \eqref{eq:pk=},
\begin{equation*}
	a_{q,\,i}=1/p_{q,\,i},\quad a_{q,\,j}=1/p_{q,\,j},\quad i,\,j=1,\,\ldots,\,n_q,
\end{equation*}
and $a_{q,\,i,\,j}$ is the inverse of the joint examination probability \eqref{eq:pij=},
\begin{equation*}
	a_{q,\,i,\,j}=1/p_{q,\,i,\,j},\quad i,\,j=1,\,\ldots,\,n_q.
\end{equation*}

We state that the following loss is unbiased,
\begin{equation}\label{eq:Lu}
		L_u = \sum_q \sum_{i,\,j=1}^{n_q} \left(\textbf{z}(f(x_{q,\,i}),\,f(x_{q,\,j}))\right)^\T\textbf{A}_{q,\,i,\,j}\q(c_{q,\,i},c_{q,\,j}).
	\end{equation}

\begin{theorem} \label{th:1}
	Under assumptions~1--2, the loss defined by \eqref{eq:A=}--\eqref{eq:Lu}, is unbiased, i.e.
	\begin{equation}\label{eq:unbiased}
		\E[L_u] = \E\left[\sum_q \sum_{i,\,j=1}^{n_q} \left(\textbf{z}(f(x_{q,\,i}),\,f(x_{q,\,j}))\right)^\T\q(r_{q,\,i},r_{q,\,j})\right].
	\end{equation}
\end{theorem}
\begin{proof}
	First note that Assumption~1 implies that for all $q\in Q$ and all $i,\,j=1,\,\ldots,\,n_q$, the terms $\q(c_{q,\,i},\,c_{q,\,j})$ and $\q(r_{q,\,i},\,r_{q,\,j})$ are related as follows,
	\begin{equation}\label{eq:qc=Bqr}
		\q(c_{q,\,i},\,c_{q,\,j}) = \textbf{B}(e_{q,\,i},\,e_{q,\,j})\cdot\q(r_{q,\,i},\,r_{q,\,j})
	\end{equation}
	with the matrix function $\textbf{B}$ defined in \eqref{eq:B=}.
	
	Next, for any $q$ and any $i,\,j=1,\,\ldots,\,n_q$ we have that
	\begin{multline}\label{eq:EB=}
		\E[\textbf{B}(e_{q,\,i},\,e_{q,\,j})\mid\mathcal{I}_q] =\\=  		\left(\begin{array}{cccc}
			p_{q,\,i,\,j} & 0 & 0 & 0 \\
			p_{q,\,i} - p_{q,\,i,\,j} & p_{q,\,i} & 0 & 0 \\
			p_{q,\,j} - p_{q,\,i,\,j} & 0, & p_{q,\,j} & 0 \\
			1 - p_{q,\,i} - p_{q,\,j} + p_{q,\,i,\,j} & 1-p_{q,\,i} & 1-p_{q,\,j} & 1.
		\end{array}\right)
	\end{multline}
	Under Assumption~2, the matrix $\textbf{A}_{q,\,i,\,j}$ given by \eqref{eq:A=} is well-defined. Combining its definition with \eqref{eq:EB=}, we compute that
	\begin{equation}\label{eq:identity}
		\textbf{A}_{q,\,i,\,j}\E[\textbf{B}(e_{q,\,i},\,e_{q,\,j})\mid\mathcal{I}_q] = \textbf{I}_4 \quad \forall\,q\in Q\quad \forall\,i,\,j=1,\,\ldots,\,n_q,
	\end{equation}
	where $\textbf{I}_4$ is the identity matrix of size 4. Therefore, for all $q\in Q$ and all $i,\,j=1,\,\ldots,\,n_q$ we have that 
	\begin{eqnarray*}
		&& \E\left[\left(\textbf{z}(f(x_{q,\,i}),\,f(x_{q,\,j}))\right)^\T\textbf{A}_{q,\,i,\,j}\q(c_{q,\,i},c_{q,\,j})\right] = \\
		&=& \E\left[\E\left[\left(\textbf{z}(f(x_{q,\,i}),\,f(x_{q,\,j}))\right)^\T\textbf{A}_{q,\,i,\,j}\q(c_{q,\,i},c_{q,\,j})\mid \mathcal{I}_q\right]\right]\\
		&=& \E\left[\left(\textbf{z}(f(x_{q,\,i}),\,f(x_{q,\,j}))\right)^\T\textbf{A}_{q,\,i,\,j}\E\left[\q(c_{q,\,i},c_{q,\,j})\mid \mathcal{I}_q\right]\right]\\
		&=& \E\Big[\left(\textbf{z}(f(x_{q,\,i}),\,f(x_{q,\,j}))\right)^\T\\
		&& {} \times \textbf{A}_{q,\,i,\,j}\E\left[\textbf{B}(e_{q,\,i},\,e_{q,\,j})\mid\mathcal{I}_q\right]\q(r_{q,\,i},\,r_{q,\,j})\Big]\\
		&\stackrel{eq.\,\eqref{eq:identity}}{=}& \E\left[\left(\textbf{z}(f(x_{q,\,i}),\,f(x_{q,\,j}))\right)^\T\q(r_{q,\,i},r_{q,\,j})\right],
	\end{eqnarray*}
	which, together with the definition of $L_u$ (see formula \eqref{eq:Lu}), implies the unbiasedness property \eqref{eq:unbiased}.
\end{proof}

The theorem above says that in expectation the loss function that we introduced by \eqref{eq:A=}--\eqref{eq:Lu} equals the loss computed from the true (unobserved) relevance indicators.

Interestingly, the matrix $\E[\textbf{B}(e_{q,\,i},\,e_{q,\,j})\mid\mathcal{I}_q]$ in \eqref{eq:EB=} admits the following interpretation. If we categorise the pair of items $(i,\,j)$ as belonging to one of the following four types: type 1 ($r_{q,\,i}=1,\,r_{q,\,j}=1$), type 2 ($r_{q,\,i}=1,\,r_{q,\,j}=0$), type 3 ($r_{q,\,i}=0,\,r_{q,\,j}=1$), or type 4 ($r_{q,\,i}=0,\,r_{q,\,j}=0$) then the element of the matrix $\E[\textbf{B}(e_{q,\,i},\,e_{q,\,j})\mid\mathcal{I}_q]$ in row $t_1$ and column $t_2$ is the probability of a pair of type $t_2$ to appear as a pair of type $t_1$ in the click feedback. It means that this matrix consists of the probabilities of pair type distortion due to the presence of position bias. The matrix $\textbf{A}_{q,\,i,\,j}$ used in the definition of the unbiased loss \eqref{eq:Lu} is the inverse of $\E[\textbf{B}(e_{q,\,i},\,e_{q,\,j})\mid\mathcal{I}_q]$ but as we explained before, it can be computed directly (without matrix inversion).

There is another observation we would like to make. The terms $\ell_{1,\,1}(f(x_{q,\,i}),\,f(x_{q,\,j}))$ and $\ell_{0,\,0}(f(x_{q,\,i}),\,f(x_{q,\,j}))$ in \eqref{eq:ell} are usually set to zero (as in \eqref{eq:ranknet}, for example). In the case of the conventional pairwise loss \eqref{eq:loss}--\eqref{eq:ell} (or, equivalently, \eqref{eq:Lc}) it means that the loss is computed only from those pairs of items in which exactly one of the items was clicked. The complexity of computing the contribution of collection $q$ to the loss is then $O(C_{q}\cdot N_{q})$ where $C_q$ and $N_q$ are the numbers of clicked and non-clicked items of collection $q$, respectively. In contrast, in the case of the unbiased loss \eqref{eq:Lu}, not all of the remaining pairs will be eliminated because of the presence of the matrix $A_{q,\,i,\,j}$ in the formula. Therefore, it is legitimate to question whether the computation of the unbiased loss has a higher complexity. Luckily, this is not the case as we will explain now. Note that all the elements of the last column of $A_{q,\,i,\,j}$ except the bottom-most element are zero. It implies that as long as $\ell_{0,\,0}(f(x_{q,\,i}),\,f(x_{q,\,j}))=0$, the pairs of items in which none of the items was clicked are still eliminated and the contribution of collection $q$ to the unbiased loss can be computed in $O(C_q^2 + C_q\cdot N_q)$ operations, which is $O(C_{q}\cdot N_{q})$ in the typical case when the number of clicked items $C_q$ is lower than the number of non-clicked items $N_q$. Thus, the complexity of computing the unbiased loss \eqref{eq:Lu} is the same as the complexity of computing its conventional (biased) counterpart.

\subsection{Estimation of Examination Probabilities}\label{subsec:estimation}
We consider the estimation of examination probabilities \eqref{eq:pk=} and \eqref{eq:pij=} as a separate big topic that is out of scope of this paper. However, in this section we briefly discuss how those probabilities can be estimated in practice.

When it comes to estimating the individual examination probabilities \eqref{eq:pk=}, multiple methods have been suggested in the literature, e.g. \cite{Wa:16,Jo:17,Fa:18,Ag:19*,Wa:18}. To the best of our knowledge, the estimation of joint examination probabilities \eqref{eq:pij=} has not been considered but as will be seen from examples in Section~\ref{sec:examples} below, in many important cases the joint examination probability $p_{q,\,i,\,j}$ can be either computed from the individual examination probabilities $p_{q,\,i}$ and $p_{q,\,j}$ or expressed in terms of the same parameters. In such cases, the estimation of the joint examination probabilities does not pose an additional problem. That being said, estimating the joint examination probabilities directly is also possible. We will demonstrate that by describing a procedure for their estimation which is analogous to \cite{Jo:17}.  

The method for the estimation of individual examination probabilities in \cite{Jo:17} assumes that the examination probability depends only on the position where the respective item was placed, that is
\begin{equation}\label{eq:theta}
	p_{q,\,i} = \theta(rank_{q,\,i})\quad \forall\,q \quad \forall\,i=1,\,\ldots,\,n_q,
\end{equation}
where $\theta$ is some function that maps ranks to the associated examination probabilities (observation propensities). Suppose that an analogous property holds for the joint examination probabilities,  i.\,e.
\begin{equation*}
	p_{q,\,i,\,j} = \psi(rank_{q,\,i},\,rank_{q,\,j})\quad \forall\,q \quad \forall\,i,\,j=1,\,\ldots,\,n_q,
\end{equation*}
with some function $\psi$ depending on a pair of ranks (positions). To estimate the value of $\psi$ for ranks $k_1$ and $k_2$ ($k_1 < k_2$) we can do the following. Whenever our existing ranker or recommender receives a user request (query) we randomly decide whether to swap the pair of items in positions $k_1$ and $k_2$ with the pair of items in positions 1 and 2. Let $i(k_1)$ and $i(k_2)$ be the items that the ranker assigns to positions $k_1$ and $k_2$, respectively. Then if we do not do the swap
\begin{equation}\label{eq:pnsw}
	P_{no-swap}\{c_{i(k_1)}c_{i(k_2)}=1\} = \psi(k_1,\,k_2) P_{no-swap}\{r_{i(k_1)}r_{i(k_2)}=1\}
\end{equation}
and if we do the swap
\begin{equation}\label{eq:psw}
	P_{swap}\{c_{i(k_1)}c_{i(k_2)}=1\} = \psi(1,\,2) P_{swap}\{r_{i(k_1)}r_{i(k_2)}=1\}.
\end{equation}
At the same time, since swapping has no effect on the intrinsic relevance of the swapped items we have that
\begin{equation}\label{eq:pnsw=psw}
	P_{no-swap}\{r_{i(k_1)}r_{i(k_2)}=1\} = P_{swap}\{r_{i(k_1)}r_{i(k_2)}=1\}.
\end{equation}
From \eqref{eq:pnsw}--\eqref{eq:pnsw=psw} we conclude that
\begin{equation}\label{eq:=fracpsi}
	\frac{P_{no-swap}\{c_{i(k_1)}c_{i(k_2)}=1\}}{P_{swap}\{c_{i(k_1)}c_{i(k_2)}=1\}} = \frac{\psi(k_1,\,k_2)}{\psi(1,\,2)}.
\end{equation}
Finally, if it is fair to assume that the first position is always examined the probability of examining both positions 1 and 2 equals the probability of examining position 2, that is
\begin{equation} \label{eq:psi12=}
	\psi(1,\,2) = \theta(2).
\end{equation}
Hence, having a consistent estimate of the individual examination probability at position 2, $\hat\theta(2)$, we can construct an estimate of the joint examination probability at positions $k_1$ and $k_2$ as follows,
\begin{equation*}
\hat\phi(k_1,\,k_2) = \hat\theta(2)\cdot\frac{\sum_{q\in Q_{no-swap}} c_{i(k_1)}c_{i(k_2)} / |Q_{no-swap}|}{\sum_{q\in Q_{swap}} c_{i(k_1)}c_{i(k_2)} / |Q_{swap}|}
\end{equation*}
which is consistent due to \eqref{eq:=fracpsi}, \eqref{eq:psi12=}, and the consistency of $\hat\theta(2)$.

Note that when the joint examination probabilities are estimated directly the number of parameters is quadratic in the number of positions but this can be mitigated by doing the estimation only for some of the position pairs and extrapolating on the rest.

\section{Examples}\label{sec:examples}

In this section, we will demonstrate how the proposed general framework can be used to derive unbiased pairwise loss functions in practice. We show by example that the framework allows us to focus on computing the examination probabilities and that once we compute them, an unbiased learning-to-rank method for the corresponding setting gets produced “automatically”.

At this point, we will make the common simplifying assumption \eqref{eq:theta}, i.\,e. suppose that the individual examination probability depends only on the position where the item is displayed. Contrary to the discussion at the end of Section~\ref{sec:framework}, we will not need to maintain a similar assumption for the joint examination probabilities because depending on the considered model of user browsing behaviour, we will be able either to compute the joint examination probabilities from the individual ones or to express them through the same parameters. We consider three concrete examples of user browsing behaviour below.

\textbf{Independent Examination.} If the examination of different positions in the layout happens independently, it holds that
\begin{equation}\label{eq:indep}
	p_{q,\,i,\,j} = \theta(rank_{q,\,i})\theta(rank_{q,\,j}).
\end{equation}
It can be easily checked that under this assumption the unbiased pairwise loss \eqref{eq:Lu} recovers the one proposed in \cite[Section~3.1]{Sa:20} (and hence, the framework from \cite{Sa:20} can be considered a special case of ours).

\textbf{Continuous Examination.} In this browsing model, users observe items continuously from positions with smaller ranks to positions with higher ranks without skipping. Eye-tracking experiments in \cite{Jo:17*} give some evidence suggesting that this may generally hold in web-search. In the case of such no-skipping behaviour,
\begin{equation*}\label{eq:cont}
	p_{q,\,i,\,j} = \min\{\theta(rank_{q,\,i}),\,\theta(rank_{q,\,j})\}.
\end{equation*}
An application of our framework \eqref{eq:A=}--\eqref{eq:Lu} gives an unbiased pairwise loss for this setting which, to the best of our knowledge, has not been proposed in the literature before. This loss can be optimised (e.g. by means of gradient descent) to obtain unbiased learning-to-rank models for the continuous examination setting.

\textbf{Row skipping.} 
Xie et al \cite{Xi:19} analyse the behaviour of users presented with a grid layout. One of the browsing models they introduce is called "row skipping". As the name suggests, it captures the tendency of users to skip over rows when going through the grid layout. Specifically, it is assumed that at the start or after each row the user skips the next row with probability $\gamma$, otherwise they browse items in the row. After examining an item at rank $v$ the user may stop browsing with probability $(1-C_v)$. In this model, the function $\theta$ that maps ranks to their associated examination probabilities takes the following form,

\begin{eqnarray}
	\theta(u) &=& \prod_{m=1}^{row(u)-1}\left((1-\gamma)\prod_{v'=S(m)+1}^{S(m)+N(m)}C_{v'} + \gamma\right)\nonumber\\&& {} \times(1-\gamma)\prod_{v=S(row(u))+1}^{u-1}C_v,\label{eq:th=}
\end{eqnarray}
where $row(u)$ is the row containing the rank $u$, $N(m)$ is the number of items in row $m$, and $S(m)$ is the total number of items before row $m$. In the above formula, the product before the $\times$ sign is the probability that the user does not quit before reaching the row containing the rank $u$. Then $(1-\gamma)$ accounts for the chance of skipping the respective row and the last product is the probability that the user does not stop before reaching the rank $u$ when examining items in the row. The joint examination probability $p_{q,\,i,\,j}$ can be expressed in a similar fashion. Specifically, denoting $\min\{rank_{q,\,i},\,rank_{q,\,j}\}$ by $h_{q,\,i,\,j}$ and $\max\{rank_{q,\,i},\,rank_{q,\,j}\}$ by $w_{q,\,i,\,j}$, we can write

\begin{eqnarray}
	p_{q,\,i,\,j} &=& \theta(h_{q,\,i,\,j})\prod_{v=h_{q,\,i,\,j}}^{S(row(h_{q,\,i,\,j}))+N(row(h_{q,\,i,\,j}))}C_v
	\nonumber\\
	&& {} \times \prod_{m=row(h_{q,\,i,\,j})+1}^{row(w_{q,\,i,\,j})-1}\left((1-\gamma)\prod_{v'=S(m)+1}^{S(m)+N(m)}C_{v'} + \gamma\right)\nonumber\\&& {} \times(1-\gamma)\prod_{v=S(row(w_{q,\,i,\,j}))+1}^{w_{q,\,i,\,j}-1}C_v,\label{eq:diffrows}
\end{eqnarray}
if items $i$ and $j$ were displayed in different rows and 

\begin{equation}\label{eq:samerow}
	p_{q,\,i,\,j} =\theta(h_{q,\,i,\,j})\prod_{v=h_{q,\,i,\,j}}^{w_{q,\,i,\,j}-1}C_v
\end{equation}
otherwise. The function $\theta$ in \eqref{eq:diffrows} and \eqref{eq:samerow} is the one from \eqref{eq:th=}.

Assuming that the parameters $\gamma$ and $C_v$ are known or have been estimated, one can plug them in \eqref{eq:th=}--\eqref{eq:samerow} and construct an unbiased pairwise loss according to \eqref{eq:A=} and \eqref{eq:Lu}. This loss can be then optimised to obtain an unbiased ranker for the case of row-skipping behaviour. This is another example of how the framework from Section~\ref{sec:framework} allows us to focus on computing the examination probabilities and an unbiased learning-to-rank method for the corresponding setting gets produced ``automatically'' - just by plugging them into \eqref{eq:A=} and \eqref{eq:Lu}.

\section{Robust Unbiased LambdaMART}\label{sec:lambda}

The unbiased loss proposed in Section~\ref{sec:framework} can be combined with the so called ``lambda-trick'' \cite{Bu:10} similarly to how it was done in \cite{Hu:19}. The lambda-trick consists of re-weighting the item pairs in the gradient of the loss function (such as \eqref{eq:Lu}) so that the optimisation process performs better at maximising an information retrieval (IR) metric (such as NDCG). The weights are set to the absolute difference $|\Delta Z_{i,\,j}|$ in the respective IR metric when the two items of the pair are swapped in the ranking induced by the current values of model parameters. The re-weighted gradient of the underlying pairwise loss function is called \emph{lambda-gradient}.

As will be seen from the derivations below, the application of the lambda-trick to the unbiased pairwise loss \eqref{eq:Lu} from Section~\ref{sec:framework} generates an interesting insight. Specifically, the resulting algorithm turns out to be the same regardless of the values of the joint examination probabilities $p_{q,\,i,\,j}$.
It means that the algorithm arising this way is valid as long as the examination hypothesis holds and the observation propensities are positive.
In particular, this is true regardless of the specific user behaviour patterns (such as the ones discussed in the previous section).
This is especially interesting given that the version of Unbiased LambdaMART stemming from \eqref{eq:Lu} is simpler than the original Unbiased LambdaMART in the sense that it does not have some of its parameters.

Similarly to \cite{Bu:10} and \cite{Hu:19}, we will set the loss values as in formula \eqref{eq:ranknet}. To write down an expression for the lambda-gradient based on our unbiased loss \eqref{eq:Lu}, we first compute the gradient of \eqref{eq:Lu} (with $\textbf{z}$ defined by \eqref{eq:ranknet}). We denote
\begin{equation*}
	\mu_{q,\,i,\,j} = \frac{1}{1 + e^{f(x_{q,\,i}) - f(x_{q,\,j})}}.
\end{equation*}
With this notation, the gradient equals
\begin{multline}\label{eq:grad}
	\left(L_u\right)'_{q,\,i} = \sum_{j=1}^{n_q}\Bigg[\\
  \left(-\sigma\mu_{q,\,i,\,j}(a_{q,\,i}-a_{q,\,i,\,j}) + \sigma\mu_{q,\,j,\,i} (a_{q,\,j}-a_{q,\,i,\,j})\right)c_{q,\,i}c_{q,\,j}\\
  {} - \sigma\mu_{q,\,i,\,j}a_{q,\,i}c_{q,\,i}(1-c_{q,\,j}) + \sigma\mu_{q,\,j,\,i}a_{q,\,j}(1-c_{q,\,i})c_{q,\,j}\Bigg].
\end{multline}
It can be seen that the gradient depends not only on pairs with different target labels but also on pairs of items that were both clicked. However, when we proceed with the lambda-trick the contribution of such pairs gets eliminated since their respective $|\Delta Z_{i,\,j}|$ is zero. This gives the following formula for the lambda-gradient.
\begin{multline}\label{eq:lambda}
	\lambda_{q,\,i} = \sum_{j=1}^{n_q}\left(
  \frac{\lambda_{q,\,i,\,j}}{p_{q,\,i}}c_{q,\,i}(1-c_{q,\,j}) - \frac{\lambda_{q,\,j,\,i}}{p_{q,\,j}}(1-c_{q,\,i})c_{q,\,j}\right),
\end{multline}
where $\lambda_{q,\,i,\,j}=-\sigma\mu_{q,\,i,\,j}$ and $\lambda_{q,\,j,\,i}=-\sigma\mu_{q,\,j,\,i}$. This version of Unbiased LabmdaMART is compared with the original LambdaMART and Unbiased LambdaMART in Table~\ref{tbl:lambda}. It can be seen that the difference between formula \eqref{eq:lambda} and the original Unbiased LambdaMART is the absence of $t^-$ parameters\footnote{In Unbiased LambdaMART \cite{Hu:19}, the parameter $t^-_{k}$ is defined as the ratio between the probability of click absence and the probability of irrelevance (at rank $k$).}.

\begin{table}
	\caption{Pair Contribution to the Lambda-Gradient (index $q$ is omitted)}\label{tbl:lambda}
	\begin{tabular}{lcc}
		\multicolumn{1}{c}{Method}   & \multicolumn{2}{c}{Pair Type}\\

		{} & $c_i>c_j$ & $c_i<c_j$ \\
		\midrule
		 LambdaMART \cite{Bu:10} & $\lambda_{i,\,j}$ & $-\lambda_{j,\,i}$\\
		 Unbiased LambdaMART \cite{Hu:19} & $\frac{\lambda_{i,\,j}}{\theta(rank_{i})t^-_{rank_{j}}}$ & $\frac{-\lambda_{j,\,i}}{\theta(rank_{j})t^-_{rank_{i}}}$\\
		 Formula \eqref{eq:lambda} & $\frac{\lambda_{i,\,j}}{\theta(rank_{i})}$ & $\frac{-\lambda_{j,\,i}}{\theta(rank_{j})}$\\
	\end{tabular}
\end{table}

The lambda-gradient formula \eqref{eq:lambda} can be alternatively viewed as if it is obtained by omitting the terms corresponding to pairs of clicked items in the gradient \eqref{eq:grad} and then applying the lambda-trick heuristic in the ``standard'' way. From this perspective, one can expect \eqref{eq:lambda} to perform better when the contribution of such pairs to the loss \eqref{eq:Lu} is small - for example, when the clicks are sparse. See also the discussion in Section~\ref{subsec:results} below.

In the next section we compare the performance of the simplified Unbiased LambdaMART given by \eqref{eq:lambda} with the original Unbiased LambdaMART in a semi-synthetic experiment.

\section{Experimental Setup and Results}\label{sec:experiments}

We performed our simulation experiments using the Yahoo! C14 Learning to Rank Challenge\footnote{\url{https://webscope.sandbox.yahoo.com/catalog.php}} dataset. This dataset consists of 29921 queries divided into three parts (train, validation, and test). Each query has an associated list of documents (of varying length). Every document is described with a feature set containing 700 features and supplied with a relevance label set by human editors \cite{CCh:10}. The relevance labels take values from 0 (\emph{irrelevant}) to 4 (\emph{highly relevant}). We used the train part for simulations and the test part for evaluation.

Our simulation setup is similar to that from \cite{Ai:18} and \cite{Hu:19}.

\subsection{Click Data Generation}

The train part of the Yahoo! C14 dataset has 19944 queries. We used the same initial rankings of the associated document lists as in \cite{Hu:19}\footnote{See \url{https://github.com/acbull/Unbiased_LambdaMart.}}. Given those initial rankings, the lists were truncated at a fixed position. We conducted several experiments with the truncation position set to 10, 20, and 30. A small number of queries that did not have any relevant documents were discarded\footnote{In our experiments we do not generate noisy clicks. Consequently, the click data generated for such queries would inevitably contain no clicks and would not be utilised by any of the methods we compare. The number of discarded queries varied between 784, 802, or 908 depending on the truncation position.}.

For each document, we generated\footnote{We used \texttt{numpy.random.default\_rng} generator with the seed set to 2022. Our source code for running the experiments is available at \url{https://github.com/zalandoresearch/pairwise-debiasing}.} a relevance indicator and an examination indicator. The click indicator was computed as the product of the two.

Relevance indicators were generated as independent binary (Bernoulli) random variables with the probability of success set to
\begin{equation*}
	P\{r_{q,\,i} = 1\} = (2^{y_{q,\,i}} - 1) / 15,
\end{equation*}
where $y_{q,\,i}$ is the manual relevance label provided in the dataset.

When generating the examination indicators we considered two user browsing models mentioned in Section~\ref{sec:examples}: continuous browsing and independent examination. In both cases, the probability of examining a given position was set as in \cite{Jo:17}, that is
\begin{equation}\label{eq:invrank}
	p_{q,\,i} = \frac{1}{rank_{q,\,i}},\quad i=1,\,\ldots,\,n_q.
\end{equation}

In the case of independent examination, the examination indicators were generated as independent Bernoulli random variables with the probability of success defined above.

In contrast, in the case of continuous browsing the examination indicators (corresponding to the same query) were dependent and had the property
\begin{equation*}
	rank_{q,\,i} \le rank_{q,\,j} \Rightarrow e_{q,\,i} \ge e_{q,\,j}\quad\forall\,q\in Q\quad\forall\,i,\,j=1,\,\ldots,\,n_q.
\end{equation*}
To achieve that we first randomly drew the last examined position $d_q\in\{1,\,\ldots,\,n_{max}\}$ where $n_{max}$ is the maximum document list length. The distribution of the last examined position $d_q$ was set to
\begin{equation*}
P\{d_q = k\} = 
	\begin{cases}
	\frac 1 {k} - \frac 1 {k+1}, & k=1,\,\ldots,\,n_{max}-1, \\
	\frac 1 {n_{max}}, & k=n_{max}.
	\end{cases}
\end{equation*}
to match \eqref{eq:invrank}. After drawing $d_q$ for each query, we set the examination indicators as follows
\begin{equation*}
	e_{q,\,i} = \begin{cases}
		1, & rank_{q,\,i} \le d_q,\\
		0, & rank_{q,\,i} > d_q
	\end{cases},\quad i=1,\,\ldots,n_q.
\end{equation*}

The generation process was repeated for each query 16 times, which formed our training data consisting of 306272 document lists with associated click indicators.

\subsection{Methods under Comparison}

In our experiments, we compared the following learning-to-rank algorithms.

\emph{LambdaMART trained on the click data.} This is LambdaMART \cite{Bu:10} fit to the click data. Similarly to \cite{Hu:19}, we consider the performance of this baseline as a lower bound. The reason is that it is trained on (biased) implicit feedback data using a machine learning algorithm that has no debiasing mechanism.

\emph{LambdaMART trained on the labeled data.} This model is LambdaMART \cite{Bu:10} trained on the ``golden'' (ground-truth) relevance labels provided in the dataset. The performance of this baseline is an upper bound since it is trained on manually labelled data, free from the position bias.

\emph{Unbiased LambdaMART.} This is the original Unbiased LambdaMART from \cite{Hu:19}. The examination propensities ($t^+$ parameters in the terminology of \cite{Hu:19}) were fixed at their true values \eqref{eq:invrank}. The estimation of $t^-$ parameters was carried out in the usual way, i.e. as part of the Unbiased LambdaMART training process. The algorithm applies additive regularisation to the $t^+$ and $t^-$ parameters (which in our case affected only the $t^-$ parameters since $t^+$ parameters were fixed at their true values). The regularisation is controlled by a hyper-parameter (denoted by $p$ in \cite{Hu:19}). We considered three values of it corresponding to no regularisation, $L_1$-regularisation, and $L_2$-regularisation, respectively.

\emph{Robust Unbiased LambdaMART.} It is the simplified version of Unbiased LambdaMART that we constructed in Section~\ref{sec:lambda}. Note that we used the true values of the observation propensities both in the simplified and in the original versions of Unbiased LambdaMART to have a fair comparison.

We did not include Regression-EM since Unbiased LambdaMART showed a better performance compared to it in the experiments from \cite{Hu:19}.

All of the models were trained using LightGBM \cite{Ke:17}. The hyper-parameters were set to the same values as in \cite{Hu:19}. In particular, the number of trees was 300, the learning rate equalled 0.05, the maximum number of leaves in a tree was 31, the feature fraction was 0.9, and the bagging fraction was set to 0.9.

\subsection{Evaluation Protocol}

We evaluated the algorithms on the test part of the dataset using the ``golden'' relevance judgements as target labels. The test part contains 6983 queries of which 248 do not have any associated documents with positive relevance labels and were excluded\footnote{The information retrieval metrics we used for evaluation are not defined for such queries.}. We did not truncate the document lists at the evaluation stage.

We evaluated the algorithms with the NDCG metric because it is the information retrieval metric we targeted when applying the lambda-trick. Specifically, we used NDCG at cutoff positions 1, 3, 5, and 10. We also report MAP for completeness.

\subsection{Experimental Results}\label{subsec:results}

The experimental results are presented in Tables~\ref{tbl:nsm} and \ref{tbl:pbm}. In each of the two tables, we report the absolute performance of unregularised Unbiased LambdaMART and the relative performance of all of the other methods (in percentages). The relative changes typeset in bold are significant at a 5\% significance level as assessed by a paired two-sided t-test with Bonferroni correction\footnote{The Bonferroni correction was applied globally, i.\,e. across all of the reported comparisons.}. We additionally checked if the performance of Robust Unbiased LambdaMART was statistically different from that of the original Unbiased LambdaMART with $L2$-regularisation. Cases where the difference is statistically significant (according to a paired two-sided t-test with a 5\% significance level) are highlighted with a frame.

Table~\ref{tbl:nsm} corresponds to the experiment with continuous examination. For that type of user behaviour, the robust version of Unbiased LambdaMART outperforms the unregularised Unbiased LambdaMART for all values of the truncation position with respect to NDCG. The uplift is especially pronounced for larger values of the truncation position (i.\,e. when the maximum training list length is bigger) and for smaller values of the cut-off position in the NDCG metric.

When the regularisation parameter in Unbiased LambdaMART gets increased the performance of it catches up with that of Robust Unbiased LambdaMART. This is expected because as the reguarisation parameter grows bigger the $t^-$ parameters in the original Unbiased LambdaMART are regularised away and the method ``converges'' to Robust Unbiased LambdaMART. However, it can be seen from Table~\ref{tbl:nsm} that the sufficient level of regularisation needed for the original Unbiased LambdaMART to perform on par with the robust version depends on the maximum length of a training list. Note that tuning the regularisation parameter on the click data using conventional validation approaches can be misleading because the click data is affected by the position bias. Instead, one would need to use an unbiased version of the validation loss (such as \eqref{eq:Lu}), similarly to \cite[Section~6.1.4]{Sa:20*}. Although the latter is a valid and feasible approach, we still consider the absence of any debiasing-related hyper-parameters in Robust Unbiased LambdaMART an advantage since it is making the method simpler.

The evaluation results for independent examination can be found in Table~\ref{tbl:pbm}. In this setting, Robust Unbiased LambdaMART performs better than the unregularised Unbiased LambdaMART for larger values of the truncation position (20 and 30). However, its performance is slightly worse than that of Unbaised LambdaMART when the maximum training list length equals 10. Our explanation is that when both the loss function from \cite{Hu:19} and the loss given by \eqref{eq:Lu} have similar unbiasedness properties (such as in the case of independent browsing) the application of the lambda-trick on top of the former can give a better result. This may be further explained by the fact that the loss from \cite{Hu:19} does not contain terms corresponding to pairs of clicked items which contribution gets eliminated by the lambda-trick. This prompts to seek an adaptation of the lambda-trick that would propagate the contribution of such pairs into the lambda-gradient. We consider this a topic for future research. Note that even in the case of independent examination the comparison outcome between the robust and the original Unbiased LambdaMART still depends on the value of the regularisation parameter for larger values of the truncation position and Unbiased LambdaMART needs to be regularised appropriately to outperform the robust version \eqref{eq:lambda}.

\begin{table*}
	\caption{Evaluation Results for the Case of Continuous Examination\\(see Section~\ref{subsec:results} for details)}\label{tbl:nsm}
	\begin{tabular}{clccccc}
		Trunc. pos. & \multicolumn{1}{c}{Method} & NDCG@1 & NDCG@3 & NDCG@5 & NDCG@10 & MAP\\
		\hline
		\hline
	\multirow{6}{*}{10} & Unbiased LambdaMART (no regularisation) & 0.684 & 0.680 & 0.702 & 0.752 & 0.880 \\
	                    & LambdaMART (click data) & {\bf -6.71\%} & {\bf -5.11\%} & {\bf -4.40\%} & {\bf -3.28\%} & {\bf -0.72\%}\\
	                    & Unbiased LambdaMART ($L_1$-regularisation) & {\bf +2.12\%} & {\bf +2.15\%} & {\bf +1.51\%} & {\bf +1.21\%} & {+0.16\%} \\
	                    & Unbiased LambdaMART ($L_2$-regularisation) & { +1.64\%} & {\bf +1.86\%} & {\bf +1.54\%} & {\bf +1.12\%} & {-0.03\%} \\
	                    & Robust Unbiased LambdaMART & { +1.01\%} & {\bf +1.75\%} & {\bf +1.38\%} & {\bf +0.96\%} & {-0.20\%}\\
	                    & LambdaMART (labelled data) & {\bf +3.93\%} & {\bf +4.41\%} & {\bf +3.65\%} & {\bf +2.82\%} & { +0.33\%} \\
	    \hline 
	 
\multirow{6}{*}{20} & Unbiased LambdaMART (no regularisation) & 0.642 & 0.651 & 0.678 & 0.732 & 0.875 \\
	                    & LambdaMART (click data) & {\bf -3.12\%} & {\bf -2.79\%} & {\bf -2.27\%} & {\bf -1.58\%} & {\bf -0.29\%}\\
	                    & Unbiased LambdaMART ($L_1$-regularisation) & {\bf +6.45\%} & {\bf +4.74\%} & {\bf +3.87\%} & {\bf +2.88\%} & {\bf +0.65\%} \\
	                    & Unbiased LambdaMART ($L_2$-regularisation) & {\bf +7.63\%} & {\bf +5.78\%} & {\bf +4.66\%} & {\bf +3.41\%} & {\bf +0.72\%} \\
	                    & Robust Unbiased LambdaMART & {\bf +8.69\%} & \framebox{\bf +6.88\%} & {\framebox{\bf +5.52\%}} & \framebox{\bf +4.09\%} & {\bf +0.74\%}\\
	                    & LambdaMART (labelled data) & {\bf +11.24\%} & {\bf +9.69\%} & {\bf +8.20\%} & {\bf +6.24\%} & {\bf +1.13\%} \\
	    \hline 

\multirow{6}{*}{30} & Unbiased LambdaMART (no regularisation) & 0.613 & 0.629 & 0.660 & 0.718 & 0.870 \\
	                    & LambdaMART (click data) & {+0.27\%} & {-0.16\%} & {-0.21\%} & {-0.06\%} & {+0.15\%}\\
	                    & Unbiased LambdaMART ($L_1$-regularisation) & {\bf +9.13}\% & {\bf +6.40\%} & {\bf +5.00\%} & {\bf +3.71\%} & {\bf +0.97\%} \\
	                    & Unbiased LambdaMART ($L_2$-regularisation) & {\bf +11.46\%} & {\bf +8.18\%} & {\bf +6.50\%} & {\bf +4.67\%} & {\bf +1.23\%} \\
	                    & Robust Unbiased LambdaMART & {\bf +13.02\%} & \framebox{\bf +10.11\%} & \framebox{\bf +7.94\%} & \framebox{\bf +5.98\%} & {\bf +1.26\%}\\
	                    & LambdaMART (labelled data) & {\bf +17.04\%} & {\bf +14.03\%} & {\bf +11.61\%} & {\bf +8.66\%} & {\bf +1.84\%} \\
	    \hline 
	\end{tabular}
		
\end{table*}

\begin{table*}
	\caption{Evaluation Results for the Case of Independent Examination\\(see Section~\ref{subsec:results} for details)}\label{tbl:pbm}
\begin{tabular}{clccccc}
		Trunc. pos. & \multicolumn{1}{c}{Method} & NDCG@1 & NDCG@3 & NDCG@5 & NDCG@10 & MAP\\
		\hline
		\hline
\multirow{6}{*}{10} & Unbiased LambdaMART (no regularisation) & 0.696 & 0.693 & 0.714 & 0.762 & 0.882 \\
	                    & LambdaMART (click data) & {\bf -7.24\%} & {\bf -5.88\%} & {\bf -5.10\%} & {\bf -3.99\%} & {\bf -0.85\%}\\
	                    & Unbiased LambdaMART ($L_1$-regularisation) & {+0.03\%} & {+0.61\%} & { +0.34\%} & {+0.16\%} & {\bf -0.28\%} \\
	                    & Unbiased LambdaMART ($L_2$-regularisation) & {-0.17\%} & {+0.44\%} & {+0.08\%} & {-0.08\%} & {\bf -0.40\%} \\
	                    & Robust Unbiased LambdaMART & {-1.37\%} & \framebox{-0.57\%} & \framebox{-0.59\%} & \framebox{\bf -0.57\%} & \framebox{\bf -0.76\%}\\
	                    & LambdaMART (labelled data) & {\bf +2.21\%} & {\bf +2.45\%} & {\bf +1.99\%} & {\bf +1.46\%} & {+0.12\%} \\
	    \hline 
	    
\multirow{6}{*}{20} & Unbiased LambdaMART (no regularisation) & 0.654 & 0.663 & 0.688 & 0.740 & 0.877  \\
	                    & LambdaMART (click data) & {\bf -2.97\%} & {\bf -3.01\%} & {\bf -2.44\%} & {\bf -1.75\%} & {\bf -0.30\%}\\
	                    & Unbiased LambdaMART ($L_1$-regularisation) & {\bf +6.45\%} & {\bf +4.80\%} & {\bf +4.12\%} & {\bf +2.98\%} & {\bf +0.48\%} \\
	                    & Unbiased LambdaMART ($L_2$-regularisation) & {\bf +6.65\%} & {\bf +5.17\%} & {\bf +4.40\%} & {\bf +3.21\%} & {\bf +0.43\%} \\
	                    & Robust Unbiased LambdaMART & {\bf +5.64\%} & {\bf +4.93\%} & {\bf +3.96\%} & {\bf +2.90\%} & \framebox{ +0.01\%}\\
	                    & LambdaMART (labelled data) & {\bf +9.09\%} & {\bf +7.57\%} & {\bf +6.54\%} & {\bf +4.98\%} & {\bf +0.89\%} \\
	    \hline

\multirow{6}{*}{30} & Unbiased LambdaMART (no regularisation) & 0.624 & 0.642 & 0.670 & 0.726 & 0.872 \\
	                    & LambdaMART (click data) & {+1.05\%} & {-0.30\%} & {-0.30\%} & {-0.14\%} & {+0.02\%}\\
	                    & Unbiased LambdaMART ($L_1$-regularisation) & {\bf +10.60}\% & {\bf +7.43\%} & {\bf +6.18\%} & {\bf +4.65\%} & {\bf +1.03\%} \\
	                    & Unbiased LambdaMART ($L_2$-regularisation) & {\bf +12.11\%} & {\bf +8.82\%} & {\bf +7.41\%} & {\bf +5.48\%} & {\bf +1.05\%} \\
	                    & Robust Unbiased LambdaMART & {\bf +11.71\%} & {\bf +8.77\%} & {\bf +7.10\%} & {\bf +5.19\%} & \framebox{\bf +0.66\%}\\
	                    & LambdaMART (labelled data) & {\bf +14.97\%} & {\bf +11.75\%} & {\bf +9.93\%} & {\bf +7.42\%} & {\bf +1.56\%} \\
	    \hline 
	\end{tabular}
		
\end{table*}

\section{Conclusion}

We advanced the theory of pairwise unbiased learning-to-rank by developing a general debiasing approach based on a minimalistic set of assumptions. We showed how our general framework can be used to construct unbiased pairwise loss functions and, consequently, unbiased learning-to-rank algorithms for different types of user behaviour. We further implemented our approach as a simplified but robust version of the Unbiased LambdaMART. Our experimental results show that this version performs better than the original algorithm when the examination of different items in the layout occurs in a dependent fashion.

One of the insights following from the theory developed in this paper is that in the presence of position bias, a learning-to-rank procedure can benefit from accounting not only for pairs with different target labels but also for pairs with the same (non-zero) target label. In the context of LambdaMART, this motivates future research aiming at adapting the lambda-trick so that it does not eliminate the contribution of such pairs to the trained model. 

Another interesting direction for future research is to combine our approach with variance reduction techniques \cite{SJo:15,SJo:15*}.

\begin{acks}
The authors would like to thank Dr. Christian Bracher from Zalando Research and Dr. Zeno Gantner for reading the draft of the paper and giving helpful feedback.
\end{acks}

\bibliographystyle{ACM-Reference-Format.bst}

\balance

\bibliography{ltr.bib}

\end{document}